\documentclass[12pt]{iopart}

\usepackage{graphicx}

\usepackage{iopams}

\newtheorem{theorem}{Theorem}

\newtheorem{corollary}{Corollary}

\newtheorem{lemma}{Lemma}

\newenvironment{proof}[1][Proof]{\textbf{#1.} }{\ \rule{0.5em}{0.5em}}

\begin{document}

\title[]{Commutativity of missing label operators in terms of Berezin brackets}

\author{Luis J. Boya\dag}
\address{\dag\ Dpto. F\'{\i}sica Te\'orica\\Facultad de Ciencias\\Universidad de Zaragoza\\
E-50009 Zaragoza, Spain} \ead{luisjo@unizar.es}

\author{Rutwig Campoamor-Stursberg\ddag}
\address{\ddag\ Dpto. Geometr\'{\i}a y Topolog\'{\i}a\\Fac. CC. Matem\'aticas\\
Universidad Complutense de Madrid\\Plaza de Ciencias, 3\\E-28040
Madrid, Spain} \ead{rutwig@mat.ucm.es}

\begin{abstract}
We obtain a criterion on the commutativity of polynomials in the
enveloping algebra of a Lie algebra in terms of an involution
condition with respect to the Berezin bracket. As an application,
it is shown that the commutativity requirement of missing label
operators for reduction chains in the missing label problem can be
solved analytically.
\end{abstract}


\maketitle

\section{Introduction}
A classical situation where group theory is applied to the
description of a physical system is concerned with classification
schemes. Here, irreducible representations of an (approximate)
symmetry Lie group must be decomposed into irreducible
representations of a certain subgroup in order to classify states:
\begin{equation}
\left|
\begin{array}
[c]{cccccccccc}%
 \frak{s} & \supset &  \frak{s}^{\prime} & \supset &  \frak{s}^{\prime\prime} & ...& \supset &  \frak{s}^{(n)} & ...\\
\downarrow &  & \downarrow &  & \downarrow & & & \downarrow & \\
 \left[  \lambda\right]   &  &  \left[
\lambda^{\prime}\right] & &  \left[ \lambda^{\prime\prime}\right]
& ...&\supset & \left[ \lambda^{(n)}\right] & ...
\end{array}
\right\rangle. \label{Red1}
\end{equation}
Depending on the situation, the labels obtained from the reduction
(\ref{Red1}) are sufficient to solve the problem if we require
multiplicity free reductions, as happens for various models in
nuclear physics \cite{Ia}. However, for non-multiplicity free
reductions, the subgroup does not always provide a sufficient
number of labels to characterize the basis states without
ambiguity. This happens in many of the non-canonical embeddings
and generic irreducible representations (IRREPs) of Lie algebras.
This is not necessarily a constraint, since often the interesting
representations belong to a certain type, and degeneracies are
solved directly with the available Casimir operators.

\medskip

Many different approaches to solve the so-called missing label
problem (short MLP) have been proposed in the literature, varying
from projection of states to the obtainment of subgroup scalars in
the enveloping algebras \cite{El,Sh2,Za,Que,Bi,Hu1}. Even if the
latter procedure allows to find the most general labelling
operator, the effective computation of integrity bases\footnote{An
integrity basis is formed by a finite number of elementary
labelling operators in terms of which all others may be written as
products.} is a rather complicated problem, and no effective
method is available. Among the difficulties appearing in this
approach, we remark that no general criterion to decide how many
 operators are necessary to generate an integrity basis is known. Further,
it should be expected that labelling operators have some
interpretation in a physical context, as happens for the Elliott
chain $\frak{su}(3)\supset \frak{so}(3)$ used in nuclear physics
or the Racah chain $\frak{so}(7)\supset G_{2}\supset\frak{so}(3)$
used in the description of  $f$-electron configurations
\cite{Ia,Is,Wy2}. Using the original conception of Lie groups as
groups of transformations with their infinitesimal generators, an
analytical approach using differential equations is possible, and
easily adaptable to the MLP \cite{Pe}. From this perspective,
labelling operators can be seen as particular solutions of a
certain subsystem of partial differential equations corresponding
to an embedded subalgebra. Classical operators are recovered
easily using the symmetrization map for tensors.

\medskip

The labelling problem is, to some extent, related to symmetry
breaking, and it seems therefore reasonable that contractions of
Lie algebras play some role \cite{Vi}. This ansatz was used in
\cite{C72} to observe that any reduction chain $\frak{s}\supset
\frak{s}^{\prime}$ is naturally related to some types of
inhomogeneous Lie algebras obtained by a contraction procedure.
This approach sufficed to generate physically interesting missing
label, and allowed to solve intrinsically some MLP, like those
with one labelling operator and others with a higher number.
Expansions of this procedure were developed in \cite{C75}, where
Casimir operators were decomposed with respect to the contraction
related to the MLP. The problem of generating commuting missing
label operators remained however open.

The main objective of this article is to enlarge the previous work
of \cite{C72,C75}, proving that commutativity of labelling
operators can be solved using an analytical approach. To this
extent, we have to consider a generalized Poisson bracket in the
space $S(\frak{g})$ of polynomial functions over the dual of a
reductive Lie algebra $\frak{g}$,\footnote{We recall that a Lie
algebra is called reductive if it is the direct sum of a
semisimple and an Abelian Lie algebra.} and some identities
developed in \cite{Ols} concerning the relation of generalized
Poisson brackets and the commutator in the enveloping algebra
$\mathcal{U}(\frak{g})$ via the standard symmetrization map. In
this sense, commuting labelling operators will correspond to
unsymmetrized labelling operators that are in involution with
respect to this generalized Poisson bracket.

\medskip

As example to the procedure, we obtain three commuting labelling
operators for the chain $\frak{sp}(6)\supset\frak{su}(3)\times
\frak{u}(1)$ used in the symplectic nuclear collective model. We
stress that, although an integrity basis for this MLP was derived
in \cite{GaK}, no explicit solution to this MLP has been
constructed.

\section{Missing label operators}

It is well known from classical theory that any semisimple Lie
algebra $\frak{g}$ possesses exactly $\mathcal{N}(\frak{g})=l$
independent Casimir operators, i.e., polynomials in the generators
that commute with all elements of the algebra, where $l$ denotes
the rank of the algebra. The eigenvalues of Casimir operators are
used to label without ambiguity the irreducible representations of
$\frak{g}$, while the states within a multiplet can be
distinguished using the Cartan subalgebra. In some situations,
however, these  operators are not enough to separate
multiplicities, and need additional operators to completely
classify states. As shown in \cite{Ra}, the total number of
internal labels required is thus
\begin{equation}
i=\frac{1}{2}(\dim \frak{g}- \mathcal{N}(\frak{g})).
\end{equation}

A similar situation holds whenever we use some subalgebra
$\frak{h}$ to label the basis states of irreducible
representations of a Lie algebra $\frak{g}$. The subgroup provides
$\frac{1}{2}(\dim \frak{h}+\mathcal{N}(\frak{h}))+l^{\prime}$
labels, where $l^{\prime}$ is the number of invariants of
$\frak{g}$ that depend only on variables of the subalgebra
$\frak{h}$ \cite{Pe}. To separate states within irreducible
representations of $\frak{g}$, we need to find
\begin{equation}
n=\frac{1}{2}\left(
\dim\frak{g}-\mathcal{N}(\frak{g})-\dim\frak{h}-\mathcal{N}(\frak{h})\right)+l^{\prime}
\label{ML}
\end{equation}
additional operators, called missing label operators. The total
number of available operators of this kind is easily shown to be
twice the number of needed labels, i.e., $m=2n$. For $n>1$, it
remains the problem of determining a set of $n$ mutually commuting
operators.

\medskip
Although in general the missing label operators are neither
invariants of the algebra nor any of its subalgebras, they can
actually be determined by means of differential equations with the
same ansatz as the general invariant problem \cite{Pe,C67,BB,AA}.

\medskip

Given a Lie algebra $ \frak{g}=\left\{X_{1},..,X_{n}\; |\;
\left[X_{i},X_{j}\right]=C_{ij}^{k}X_{k}\right\}$ in terms of
generators and commutation relations, classical Casimir operators
are polynomials $C_{p}=\alpha^{i_{1}..i_{p}}X_{i_{1}}..X_{i_{p}}$
in the generators of $\frak{s}$ such that the constraint $
\left[X_{i},C_{p}\right]=0$,\; ($i=1..n$) is satisfied. Therefore
they are elements of the centre of the enveloping algebra
$\mathcal{U}(\frak{g})$ of $\frak{g}$. In order to cover arbitrary
Lie groups, it is convenient to use an analytical realization. The
generators of the Lie algebra $\frak{g}$ are realized in the space
$C^{\infty }\left( \frak{g}^{\ast }\right) $ by means of the
differential operators:
\begin{equation}
\widehat{X}_{i}=C_{ij}^{k}x_{k}\frac{\partial }{\partial x_{j}},
\label{Rep1}
\end{equation}
where $\left\{ x_{1},..,x_{n}\right\}$ are the components of a
covector in a dual basis of $\left\{X_{1},..,X_{n}\right\} $. The
invariants of $\frak{g}$ (in particular, the Casimir operators)
are then solutions of the following system of partial differential
equations:
\begin{equation} 
\widehat{X}_{i}F=0,\quad 1\leq i\leq n.  \label{sys}
\end{equation}
For a polynomial solution of (\ref{sys}), the standard
symmetrization map defined by
\begin{equation}
\Lambda\left(x_{i_{1}}..x_{i_{p}}\right)=\frac{1}{p!}\sum_{\sigma\in
S_{p}}X_{\sigma(i_{1})}..X_{\sigma(i_{p})}\label{Sim1}
\end{equation}
allows to recover the Casimir operators in their usual form, i.e,
as elements in the centre of $\mathcal{U}(\frak{g})$. The number
$\mathcal{N}(\frak{g})$ of functionally independent solutions of
(\ref{sys}) is obtained from the classical criteria :
\begin{equation}
\mathcal{N}(\frak{g}):=\dim \,\frak{g}- {\rm rank}\left(
C_{ij}^{k}x_{k}\right), \label{BB}
\end{equation}
where $A(\frak{g}):=\left(C_{ij}^{k}x_{k}\right)$ is the matrix
associated to the commutator table of $\frak{g}$ over the given
basis.

\medskip

If we now consider an algebra-subalgebra chain $\frak{s}\supset\frak{s}%
^{\prime}$ determined by an embedding
$f:\frak{s}^{\prime}\rightarrow \frak{s}$, the missing label
operators can be computed considering the equations of (\ref{sys})
corresponding to the generators of the subalgebra
$\frak{s}^{\prime}$. This system, as proven in \cite{Pe}, has
exactly $\mathcal{N}(f(\frak{s}^{\prime}))=\dim \frak{s}-\dim
\frak{s}^{\prime}-l^{\prime}$ solutions. Using formula (\ref{ML})
it follows further that $\mathcal{N}(f(\frak{s}^{\prime}))$ can be
expressed in terms of the number of invariants of the
algebra-subalgebra chain:
\begin{equation}
\mathcal{N}(f(\frak{s}^{\prime}))=m+\mathcal{N}(\frak{s})+\mathcal{N}(
\frak{s}^{\prime})-l^{\prime}. \label{ML2}
\end{equation}
This shows that the differential equations corresponding to the
subalgebra generators have exactly $n$ more solutions as needed to
solve the missing label problem, as expected. Even using the
analytical approach, to find a complete set of solutions for the
labelling problem is a non-trivial task.

\section{Commuting polynomials in enveloping algebras}

In this section we show a commutation property of polynomials in
enveloping algebras that will help us later to find an analytical
criterion to check whether two given labelling operators
commute.\newline Let $\frak{g}$ be a Lie algebra and
$\mathcal{U}\left( \frak{g}\right)  $ its enveloping algebra.
Further let $S\left( \frak{g}\right)  $ denote the space of
polynomials defined on the dual space $\frak{g}^{\ast}$. The
standard symmetrization of monomials (\ref{Sim1}) can be easily
extended to a one-to-one linear map
\begin{equation}
\Lambda:S\left(  \frak{g}\right)  \rightarrow\mathcal{U}\left(  \frak{g}%
\right).\label{Sim2}
\end{equation}
In particular, if $x_{k_{1}}..x_{k_{p}}$ is a monomial of degree
$p$, then
for any $\sigma\in S_{p}$ we have the identity%
\begin{equation*}
\Lambda\left(  x_{k_{1}}..x_{k_{p}}\right)  =\Lambda\left(
x_{\sigma\left( k_{1}\right)  }..x_{\sigma\left(  k_{p}\right)
}\right)  .\label{Sim4}
\end{equation*}
Conversely, given a polynomial
$P=c^{k_{1}..,k_{n}}X_{k_{1}}..X_{k_{n}} \in\mathcal{U}\left(
\frak{g}\right)  $, we find its analytical counterpart
$\pi(P)=P^{\prime}=c^{k_{1}..,k_{n}}x_{k_{1}}..x_{k_{n}}$ by
simply replacing the generator $X_{i}$ by the corresponding
variable $x_{i}$ of $\frak{g}^{\ast}$.

\smallskip

It is well known that $\mathcal{U}\left(  \frak{g}\right)  $ is
naturally filtered \cite{Di}. For any positive integer $n$, let
$\mathcal{U}_{n}\left( \frak{s}\right)  $ be the subspace of
$\mathcal{U}\left( \frak{g}\right)  $ generated by the products
$X_{1}...X_{p}$, where $p\leq n$. Clearly $\mathcal{U}_{n}\left(
\frak{g}\right) \subset\mathcal{U}_{n+1}\left(  \frak{g}\right) $
for all
$n$, and $\mathcal{U}\left(  \frak{g}\right)  =\bigcup_{n}\mathcal{U}%
_{n}\left(  \frak{g}\right)  $. Given an element
$u\in\mathcal{U}\left(\frak{g}\right)  $, we say that it has filtration $p$ if $u\in\mathcal{U}%
_{p}\left(  \frak{g}\right)  $ and $u\notin\mathcal{U}_{p-1}\left(
\frak{g}\right)  $.

\smallskip

It should be observed that, even if $\pi\left(  u\right)
=\pi\left( v\right)  $, the elements $u$ and $v$ do not
necessarily commute. Taking e.g., $u=X_{1}X_{2}$ and
$v=X_{2}X_{1}$, a routine verification shows that
\begin{eqnarray*}
\fl \left[  u,v\right] =-\left[  X_{2},X_{1}\right]
X_{2}X_{1}-\left[ X_{1},X_{2}\right]
X_{2}X_{1}-\left[  X_{2},\left[  X_{1},X_{2}\right]  \right]  X_{1}%
-X_{2}\left[  X_{1},\left[  X_{2},X_{1}\right]  \right]  \\
\lo= -\left[  X_{2},\left[  X_{1},X_{2}\right]  \right]
X_{1}-X_{2}\left[ X_{1},\left[  X_{2},X_{1}\right]  \right]  .
\end{eqnarray*}

This is an immediate consequence of the fact that the variables
appearing are the same, but in reverse order. This situation can
be generalized in straightforward manner to an arbitrary number of
generators, covering also the case where $\pi\left( u\right) \neq
\pi\left(  v\right)  $.

\begin{lemma}
Let $P=X_{1}..X_{r}P_{1}\in U_{p}\left(  \frak{g}\right)  ,\;Q=X_{r}%
..X_{1}Q_{1}\in U_{q}\left(  \frak{g}\right)  $ such that $\left[  P_{1}%
,Q_{1}\right]  =0$ and $\left[  X_{i},P_{1}\right]  =\left[  X_{i}%
,Q_{1}\right]  =0$ for any $r=1,..,l.$ Then the commutator
satisfies $\left[ P,Q\right] \in U_{p+q-2}\left( \frak{g}\right)$
and $\left[ P,Q\right] \notin U_{p+q-1}\left( \frak{g}\right)$.
\end{lemma}

As remarked, this special case is not very surprising, since
products of this kind are degenerate in some sense. Actually,
reordering the first $r$ generators of $Q$, we can rewrite this
polynomial as
$Q=X_{1}..X_{r}Q_{1}-\sum_{k,l=1..r}X_{r}...\widehat{X}_{k}\widehat{X}_{l}\left[
X_{l},X_{k}\right]...X_{1}Q_{1}$. This shows at once that the
commutator reduces to
$\left[P,Q\right]=-\sum_{k,l=1..r}\left[X_{1}..X_{r},
X_{r}...\widehat{X}_{k}\widehat{X}_{l}\left[
X_{l},X_{k}\right]...X_{1}\right]P_{1}Q_{1}$, thus $\left[
P,Q\right] \in U_{p+q-2}\left(  \frak{g}\right)$. Observe in
particular that, in this case, we have that the projections
$\pi(P)$ and $\pi(Q)$ share a common factor. For the situation
that interests us, however, such products are very unlikely to
appear, thus in the following we exclude them of
our analysis. We shall say that a pair of monomials $P=X_{i_{1}%
}...X_{i_{p}}\in\mathcal{U}_{p}\left(  \frak{g}\right)$, $Q=X_{j_{1}%
}...X_{j_{q}}\in\mathcal{U}_{q}\left(  \frak{g}\right)  $ such
that $\pi\left(  P\right)  \neq\pi\left(  Q\right)  $ is
factorizable if they can be written in the form
\begin{equation}
P=X_{1}^{a_{1}}..X_{l}^{a_{l}}P_{1}\in U_{p}\left( \frak{s}\right)
,\;Q=X_{l}^{a_{l}}..X_{1}^{a_{1}}Q_{1}\in U_{q}\left(
\frak{s}\right),\label{FAK}
\end{equation}
where $\left[  P_{1},Q_{1}\right]  =0$ and $\ \left[
X_{i},P_{1}\right] =\left[  X_{i},Q_{1}\right]  =0,\;i=1,..,l$.
The pair $P,Q$ is non-factorizable if no decomposition of the
preceding type exists.

\begin{lemma}
Let $P=X_{i_{1}}...X_{i_{p}}\in\mathcal{U}_{p}\left(
\frak{g}\right)  $ and
$Q=X_{j_{1}}...X_{j_{q}}\in\mathcal{U}_{q}\left(  \frak{g}\right)
$ be non-factorizable pair. If $\left[  P,Q\right]\neq 0$, then
$\left[  P,Q\right]$ has at least one term of filtration $p+q-1$.
\end{lemma}

\begin{proof}
By the properties of the filtration in enveloping algebras, we
have that $PQ$ and $QP$ have filtration $p+q$, and further that
\[
\left[  P,Q\right]  =X_{i_{1}}...X_{i_{p}}X_{j_{1}}...X_{j_{q}}-X_{j_{1}%
}...X_{j_{q}}X_{i_{1}}...X_{i_{p}}\in\mathcal{U}_{p+q-1}\left(  \frak{g}%
\right)  .
\]
Expanding this product with the usual rule, we arrive at a sum of the type%
\begin{equation}
\left[  P,Q\right]  =-\sum_{a=1..q;\;b=1..p}X_{i_{1}}...\widehat{X}_{i_{b}%
}X_{j_{1}}..\widehat{X}_{j_{a}}\left[  X_{j_{a}},X_{i_{b}}\right]
...X_{j_{q}}X_{i_{b+1}}...X_{i_{p}}.\label{KL1}%
\end{equation}
By assumption, after reordering the sum using the
Poincar\'e-Birkhoff-Witt theorem, all terms of filtration $p+q-1$
must cancel. We argue basing on the
number of different generators in the product $PQ=X_{i_{1}}...X_{i_{p}%
}X_{j_{1}}...X_{j_{q}}$ (which coincides with that of $QP$).
Suppose there are exactly $k\leq p+q$ different generators
$X_{1},..,X_{k}$ \ in $PQ$. We further define $m_{t}$ as the
multiplicity of the generator $X_{t}$ in the product $PQ$. Then
any generator appearing in $QP$ and $PQ$ equals some $X_{t}$ for
$t=1,..,k$. It is obvious that if $X_{j_{b}}=X_{i_{a}}$, then
$\left[  X_{j_{b}},X_{i_{a}}\right]  =0$ and the term
$X_{i_{1}}...\widehat
{X}_{i_{b}}X_{j_{1}}..\widehat{X}_{j_{a}}\left[
X_{j_{a}},X_{i_{b}}\right] ...X_{j_{q}}X_{i_{b+1}}...X_{i_{p}}$
vanishes. Thus in order to evaluate the commutator $\left[
P,Q\right]  $, we only have to consider the products $\left[
X_{j_{a}},X_{i_{b}}\right]  $ $\neq0$. We now group together all
terms involving generators with the same multiplicity, i.e., for
any $X_{\alpha},X_{\beta}\in\left\{  X_{1},..,X_{k}\right\}  $
with $X_{\alpha }\neq X_{\beta}$ we define
\begin{equation}
\Phi_{\lbrack\alpha,\beta]}=\sum_{j_{a},i_{b}}X_{i_{1}}...\widehat{X}_{i_{b}%
}X_{j_{1}}..\widehat{X}_{j_{a}}\left[  X_{j_{a}},X_{i_{b}}\right]
...X_{j_{q}}X_{i_{b+1}}...X_{i_{p}},\label{KL2}
\end{equation}
where either $X_{j_{a}}=X_{\alpha},\;X_{i_{b}}=X_{\beta}$ or $X_{j_{a}%
}=X_{\beta},\;X_{i_{b}}=X_{\alpha}$. With the help of these
polynomials, the commutator $\left[ P,Q\right]  $ can be expressed
as $\left[P,Q\right]  =-\sum_{1\leq\alpha,\beta\leq
k}\Phi_{\lbrack\alpha,\beta]}$. Reordering each term $\Phi_{\left[
\alpha,\beta\right]  }$ with respect to some basis of $\frak{g}$,
$\Phi_{\left[ \alpha,\beta\right]  }$ can be rewritten as
\begin{equation}
\Phi_{\lbrack\alpha,\beta]}=\lambda_{\left(  \alpha,\beta\right)  }\,X_{i_{1}%
}...\widehat{X}_{i_{b}}X_{j_{1}}..\widehat{X}_{j_{a}}\left[  X_{j_{a}%
},X_{i_{b}}\right]
...X_{j_{q}}X_{i_{b+1}}...X_{i_{p}}+L.O.T.,\label{FUN}
\end{equation}
where $\lambda_{\left(  \alpha,\beta\right)  }\geq0$ is an
integer. The only possibility that $\lambda_{\left(
\alpha,\beta\right)  }=0$ occurs is that the commutator $\left[
X_{j_{a}},X_{i_{b}}\right]  $ appears with the same multiplicity
as $\left[  X_{i_{b}},X_{j_{a}}\right]$. This means that the
generators $X_{j_{a}},X_{i_{b}}$ appear in both $P$ and $Q$ with
the same multiplicity, but in reverse order. By assumption,
$\pi\left(  P\right)
\neq\pi\left(  Q\right)  $, thus there exists at least one generator $X_{l}%
\in\left\{  X_{1},..,X_{k}\right\}  $ such that its multiplicity
in $P$ is different from that in $Q$.\footnote{This also covers
the case where a generator appears in $P$ and not in $Q$, or
conversely. } This ensures the existence of one pair of indices
$\alpha,\beta$ such that $\lambda_{\left( \alpha,\beta\right)
}>0$. In this case, the term of filtration $p+q-1$ does not vanish
unless $\left[  X_{j_{a}},X_{i_{b}}\right]  =0$. Discarding the
zero terms, $\left[  P,Q\right]  $ is expressed as follows as a
linear combination of basis elements:
\begin{equation}
\fl \left[  P,Q\right]  =-\sum_{1\leq\alpha,\beta\leq
k}\lambda_{\left( \alpha,\beta\right)
}\,X_{i_{1}}...\widehat{X}_{i_{b}}X_{j_{1}}..\widehat
{X}_{j_{a}}\left[  X_{j_{a}},X_{i_{b}}\right]  ...X_{j_{q}}X_{i_{b+1}%
}...X_{i_{p}}+L.O.T.,\label{FUN2}
\end{equation}
where $\lambda_{\left( \alpha,\beta\right)  }\neq0$. Observe now
that if for any pair $\left\{  X_{\alpha},X_{\beta}\right\}  $
such that $\lambda_{\left( \alpha,\beta\right)  }\neq0$ we have
that $\left[ X_{\alpha},X_{\beta }\right]  =0$, then $P$ and $Q$,
after some reordering, would admit a decomposition of type
(\ref{FAK}), and therefore be a factorizable pair, which
contradicts the initial assumption on their structure. Thus there
exists at least a pair of indices $\alpha,\beta$ in $\left\{
1,..,k\right\} $ such that $\lambda_{\left(  \alpha,\beta\right)
}\neq0$ and $\left[ X_{\alpha},X_{\beta}\right]  \neq0$.

If $\left[  P,Q\right]  $ is a sum of products of at most $p+q-2$
elements, then the basis elements
$X_{i_{1}}...\widehat{X}_{i_{b}}X_{j_{1}}..\widehat
{X}_{j_{a}}\left[  X_{j_{a}},X_{i_{b}}\right]  ...X_{j_{q}}X_{i_{b+1}%
}...X_{i_{p}}$ of (\ref{FUN2}) give rise to the dependence
relation
\begin{equation}
-\sum_{1\leq\alpha,\beta\leq k}\lambda_{\left( \alpha,\beta\right)
}\,X_{i_{1}}...\widehat{X}_{i_{b}}X_{j_{1}}..\widehat{X}_{j_{a}}\left[
X_{j_{a}},X_{i_{b}}\right]
...X_{j_{q}}X_{i_{b+1}}...X_{i_{p}}=0.\label{FUN3}
\end{equation}
We observe that each term of this sum has, leaving aside the
commutator $\left[  X_{j_{a}},X_{i_{b}}\right]  $, multiplicity
$m_{t}$ in the generators $X_{t}\neq\left\{
X_{\alpha},X_{\beta}\right\}  $, multiplicity $m_{\alpha }-1$ in
$X_{\alpha}$ and $m_{\beta}-1$ in $X_{\beta}$, respectively. Let
$\left[  X_{j_{a}},X_{i_{b}}\right]
=C_{j_{a}i_{b}}^{t}X_{t}\in\frak{g}$.
If there exists some $t_{0}$ such that $X_{t_{0}}\notin\left\{  X_{1}%
,..,X_{k}\right\}  $, then the term $\lambda_{\left(
\alpha,\beta\right)
}C_{j_{a}i_{b}}^{t_{0}}X_{i_{1}}...\widehat{X}_{i_{b}}X_{j_{1}}..\widehat
{X}_{j_{a}}X_{t_{0}}...X_{j_{q}}X_{i_{b+1}}...X_{i_{p}}$ has
multiplicity one in the generator $X_{t_{0}}$, and it follows at
once by the multiplicity in the generators $\left\{
X_{1},..,X_{k}\right\}  $ that this term cannot \ be compensated
with the remaining terms of (\ref{FUN3}). If $X_{t}\in\left\{
X_{1},..,X_{k}\right\}  $ for any $t$, then two possibilities are
given:

\begin{enumerate}
\item $X_{t_{0}}=X_{l}$ for some $l\neq\alpha,\beta$. Here
$\lambda_{\left(  \alpha,\beta\right)  }C_{j_{a}i_{b}}^{t_{0}}X_{i_{1}%
}...\widehat{X}_{i_{b}}X_{j_{1}}..\widehat{X}_{j_{a}}X_{t_{0}}...X_{j_{q}%
}X_{i_{b+1}}...X_{i_{p}}$ has multiplicity $m_{t}$ in the
generators $X_{t}\neq\left\{  X_{\alpha},X_{\beta},X_{l}\right\}
$, multiplicity $m_{\alpha}-1$ in $X_{\alpha}$,. $m_{\beta}-1$ in
$X_{\beta}$ and $m_{l}+1$ in $X_{l}$, respectively. No other term
in (\ref{FUN3}) has the same multiplicities, thus $\left[
P,Q\right] \notin\mathcal{U}_{p+q-2}\left(  \frak{g}\right)$.

\item $X_{t_{0}}=X_{j_{a}}$ ( or $X_{t_{0}}=X_{i_{b}}$).
In this case,  $\lambda_{\left(  \alpha,\beta\right)  }C_{j_{a}i_{b}}^{t_{0}%
}X_{i_{1}}...\widehat{X}_{i_{b}}X_{j_{1}}..\widehat{X}_{j_{a}}X_{t_{0}%
}...X_{j_{q}}X_{i_{b+1}}...X_{i_{p}}$ has multiplicity $m_{t}$ in
the generators $X_{t}\neq\left\{  X_{\beta}\right\}  $ and
multiplicity $m_{\beta}-1$ in $X_{\beta}$. Again, it is not
possible that this term cancels with the remaining basis elements
of (\ref{FUN3}), since their multiplicities in the generators are
not the same.
\end{enumerate}
Therefore, if the commutator does not vanish, there is at least
one term of filtration $p+q-1$ that does not cancel.
\end{proof}

\section{Berezin bracket}

In this section, using an an important relation between the
commutator in the enveloping algebra $\mathcal{U}\left(
\frak{g}\right)  $ of a Lie algebra $\frak{g}$ and a generalized
Poisson bracket in $S(\frak{g})$, we show that the commutativity
requirement of labelling operators can be expressed, in the
commutative frame, by means of an involution condition with
respect to a special generalized Poisson bracket. This enables us
to find an analytical characterization for missing label operators
to commute mutually.

\smallskip

Now let $f=c^{k_{1}..,k_{p}}x_{k_{1}%
}..x_{k_{p}}\in S\left(  \frak{g}\right)  $ be a given a
homogeneous polynomial. It follows from (\ref{Sim2}) that writing
its image $\Lambda\left( f\right)$ in symmetric form is unique.
Moreover, choosing an arbitrary permutation $\sigma\in S_{p}$, we
have that
\begin{eqnarray}
\Lambda\left(  f\right)    & =\frac{1}{p!}\sum_{\nu\in\mathcal{S}_{p}%
}c^{k_{1}..,k_{p}}X_{\nu\left(  k_{1}\right) }..X_{\nu\left(
k_{p}\right)  }=\left(  c^{k_{1}..,k_{p}}X_{k_{1}}..X_{k_{p}%
}+L.O.T.\right)  \nonumber\\
& =\left(  c^{k_{1}..,k_{n}}X_{\sigma\left( k_{1}\right)
}..X_{\sigma\left(  k_{p}\right)  }+L.O.T.\right) ,\label{Sim3}
\end{eqnarray}
where L.O.T. designates the terms of order $q\leq p-1$.

\bigskip

In \cite{Ber} it was shown that for any $X_{i}\in\frak{g}$ and any
symmetrized polynomial $\Lambda(f)$, the commutator
$\left[X_{i},\Lambda(f)\right]$ in $\mathcal{U}(\frak{g})$
corresponds, under the symmetrization map $\Lambda$, to a linear
operator $A_{x_{i}}(f)$ in $S(\frak{g})$, where
$\Lambda(x_{i})=X_{i}$. These operators enabled to solve the
problem of finding a polynomial in $S(\frak{g})$, the
symmetrization $\Lambda(f)$ of which coincides with the commutator
in $\mathcal{U}(\frak{g})$ of two previously given symmetrized
polynomials $g$ and $h$. More specifically, given two given
(homogeneous) polynomials $g,h\in S(\frak{g})$, the commutator
$\left[ \Lambda\left( g\right) ,\Lambda\left( h\right) \right]
=\Lambda\left( g\right) \Lambda\left(  h\right) -\Lambda\left(
h\right) \Lambda\left( g\right)  $ of their symmetrizations is the
image, under $\Lambda$, of some polynomial $f\in
S(\frak{g})$.\footnote{Expressed in brackets, this condition is
given by $\left[ \Lambda\left( g\right) ,\Lambda\left( h\right)
\right]=\Lambda(f)$.} A routine but long and tedious computation
\cite{Ber} shows that the latter polynomial $f\in S\left(
\frak{g}\right)$ is given by the expression
\begin{equation}
\fl f=-C_{ij}^{k}x_{k}\frac{\partial g}{\partial
x_{i}}\frac{\partial h}{\partial
x_{j}}+F\left(x_{k}, \frac{\partial^{2}g}{\partial x_{i}\partial x_{j}}%
,\frac{\partial^{2}h}{\partial x_{j}\partial
x_{i}},...,\frac{\partial^{p}g}{\partial x_{j_{1}}..\partial x_{j_{p}}}%
,\frac{\partial^{p}h}{\partial x_{j_{1}}..\partial
x_{j_{p}}},...\right),\label{BeF}
\end{equation}
where $F$ is a polynomial the terms of which involve derivatives
of order $d\geq 2$ of $g$ and $h$ [(\ref{Ber}), equation (31)]. In
particular, if $g,h$ are homogeneous of degrees $p$ and $q$,
respectively, then $F$ can be decomposed as a sum of homogeneous
polynomials of degrees $\leq p+q-2$ \cite{Ber}.

\medskip

Formula (\ref{BeF}) already generalizes, in some sense, the
analytical approach to compute Casimir operators of Lie algebras
\cite{Pe}. Observe that taking $g=x_{l}$, the preceding expression
reduces to
\begin{equation}
f=-C_{ij}^{k}x_{k}\frac{\partial x_{l}}{\partial
x_{i}}\frac{\partial h}{\partial
x_{j}}=-C_{lj}^{k}x_{k}\frac{\partial h}{\partial
x_{j}}=\widehat{X}_{l}(h), \label{BeFa}
\end{equation}
where $F$ vanishes because of $\frac{\partial^{p} x_{l}}{\partial
x_{j_{1}}..x_{j_{p}}}=0$ for any $p\geq 2$. In particular, if $h$
is a Casimir invariant of $\frak{g}$, then $\widehat{X}_{l}(h)=0$,
and by the symmetrization map $\Lambda$ we obtain that
$\left[X_{l},\Lambda(h)\right]=0$.\footnote{This identification
has already been used in the frame of completely integrable
Hamiltonian systems \cite{Per}.} This fact suggests that
(\ref{BeF}) can be used to obtain an analytical criterion for the
commutativity of labelling operators.

\medskip

If we restrict to the important case where $g,h$ are homogeneous
polynomials, then equation (\ref{BeF}) can be rewritten as
\begin{equation}
\left[  \Lambda\left(  g\right)  ,\Lambda\left(  h\right)  \right]
=\Lambda\left(  \left\{  g,h\right\}  \right)  +{\rm
L.O.T.},\label{Olsh}
\end{equation}
where
\begin{equation}
\left\{  g,h\right\}  =-C_{ij}^{k}x_{k}\frac{\partial g}{\partial x_{i}}%
\frac{\partial h}{\partial x_{j}}\label{Ber}
\end{equation}
is called the Berezin bracket of $g$ and $h$ and the lower order
terms correspond to the symmetrization of the polynomial $F$
\cite{Ols,Ber}. It can further be shown that the Berezin bracket
is a generalized Poisson bracket on the space $S(\frak{g})$
\cite{Ols,Per}. We observe that equation (\ref{Olsh}) is valid for
more general Poisson brackets \cite{Ols}.

\medskip

Let $F=c^{k_{1}..,k_{p}}X_{k_{1}}..X_{k_{p}}$,
$G=c^{j_{1}..,j_{q}}X_{j_{1} }..X_{j_{q}}$ be two polynomials in
the generators of $\frak{g}$. Generalizing the situation of the
previous section, we will say that $F,G$ forms a non-factorizable
pair if for any pair
$\left\{c^{k_{1}..,k_{p}},\;c^{j_{1}..,j_{q}}\right\}$ the
monomials $\left\{  X_{k_{1}}..X_{k_{p}},\;X_{j_{1}}..X_{j_{q}%
}\right\}$ do not admit a decomposition of type (\ref{FAK}).

\begin{theorem}
Let $F,G$ be a non-factorizable pair of polynomials in the
enveloping algebra $\mathcal{U}(\frak{g})$ of $\frak{g}$ such that
$F=\Lambda(f),\; G=\Lambda(g)$ for some homogeneous polynomials
$f,g\in S(\frak{g})$. Then $\left[ F,G\right]=0$ if and only if
$\left\{ f,g\right\} =0$, i.e., if the functions $f,g$ are in
involution with respect to the Berezin bracket.
\end{theorem}

\begin{proof}
Let $f=a^{k_{1}..,k_{p}}x_{k_{1}} ..x_{k_{p}}$ and $g
=b^{j_{1}..,j_{q}}x_{j_{1}}..x_{j_{q}}$ be homogeneous polynomials
of degrees $p$ and $q$, respectively, and let $\Lambda\left(
f\right)  =\frac{1}{p!}\sum_{\sigma\in
\mathcal{S}_{p}}c^{k_{1}..,k_{p}}X_{\sigma\left(  k_{1}\right)  }%
..X_{\sigma\left(  k_{p}\right)  }$ and $\Lambda\left(  g\right)
=\frac
{1}{q!}\sum_{\sigma\in\mathcal{S}_{q}}c^{j_{1}..,j_{q}}X_{\sigma\left(
j_{1}\right)  }..X_{\sigma\left(  j_{q}\right)  }$ be their
symmetrizations. Suppose that their commutator in
$\mathcal{U}(\frak{g})$ vanishes:
\begin{equation}
\left[  \Lambda\left(  f\right)  ,\Lambda\left(  g\right)  \right]
=0.\label{Kom1}
\end{equation}
By homogeneity, formula (\ref{Olsh}) gives rise to the identity
\begin{equation}
\left[  \Lambda\left(  f\right)  ,\Lambda\left(  g\right)  \right]
=\Lambda\left(  \left\{  f,g\right\}  \right)  +L.O.T.=0
\end{equation}
Since the Berezin bracket $\left\{  f,g\right\}
=C_{ij}^{k}x_{k}\frac{\partial f}{\partial x_{i}}\frac{\partial
g}{\partial x_{j}}$ is a homogeneous polynomial of degree $p+q-1$,
the symmetric representative of $\Lambda\left(  \left\{
f,g\right\}  \right)  $ is unique,\footnote{It is at this step
where the homogeneity is essential for the validity of the
argument, since equation (\ref{Olsh}) does not necessarily hold if
homogeneity is not given.} thus the vanishing of the commutator
(\ref{Kom1}) implies in particular that $\Lambda\left( \left\{
f,g\right\} \right)=0$, because of the injectivity of
(\ref{Sim2}), and therefore that
\begin{equation}
\left\{  f,g\right\}  =0.
\end{equation}
This proves that $f$ and $g$ are in involution with respect to the
Berezin bracket.
\medskip
Conversely, let  $f$ and $g$ be homogeneous polynomials such that
the constraint
\begin{equation}
\left\{  f,g\right\}  =C_{ij}^{k}x_{k}\frac{\partial f}{\partial x_{i}}%
\frac{\partial g}{\partial x_{j}}=0
\end{equation}
holds. It follows from the definition of the symmetrization map
(\ref{Sim2}) that for any monomial $x_{1}...x_{k}$
\begin{equation}
\Lambda\left(  x_{1}...x_{k}\right)  =X_{1}...X_{k}+L.O.T.,
\end{equation}
This means that we can rewrite the symmetrized polynomials
$\Lambda(f)$ and $\Lambda(g)$ as
\begin{eqnarray}
\Lambda\left(  f\right)    & =\frac{1}{p!}\sum_{\sigma\in\mathcal{S}_{p}%
}a^{j_{1}..,j_{p}}X_{\sigma\left(  j_{1}\right)
}..X_{\sigma\left(
j_{p}\right)  }=\left(  a^{j_{1}..,j_{p}}X_{j_{1}}..X_{j_{p}%
}+L.O.T.\right), \nonumber \\
\Lambda\left(  g\right)    & =\frac{1}{q!}\sum_{\sigma\in\mathcal{S}_{q}%
}b^{k_{1}..,k_{q}}X_{\sigma\left(  k_{1}\right)
}..X_{\sigma\left(
k_{q}\right)  }=\left(  b^{k_{1}..,k_{q}}X_{k_{1}}..X_{k_{q}%
}+L.O.T.\right).\label{Sim5}
\end{eqnarray}
Evaluating the commutator in the enveloping algebra gives
\begin{eqnarray}
\left[  \Lambda\left(  f\right)  ,\Lambda\left(  g\right) \right]
&
=\left[  a^{j_{1}..,j_{p}}X_{j_{1}}..X_{j_{p}}+L.O.T.,b^{k_{1}%
..,k_{q}}X_{k_{1}}..X_{k_{q}}+L.O.T.\right],\nonumber  \\
& =\left(  a^{j_{1}..,j_{p}}b^{k_{1}..,k_{q}}\left[  X_{j_{1}%
}..X_{j_{p}},X_{k_{1}}..X_{k_{q}}\right]  +L.O.T.\right).
\end{eqnarray}
By formula (\ref{Olsh}), the commutator has no terms of degree
$p+q-1$, and since $\left\{X_{j_{1}}..X_{j_{p}},X_{k_{1}
}..X_{k_{q}}\right\}$ is a non-factorizable pair, it follows from
Lemma 2 that
\begin{equation}
a^{j_{1}..,j_{p}}b^{k_{1}..,k_{q}}\left[  X_{j_{1}}..X_{j_{p}},X_{k_{1}%
}..X_{k_{q}}\right]  =0.\label{Sim6}
\end{equation}
Now observe that, since $\Lambda\left(  f\right)  $ and
$\Lambda\left( g\right)  $ are the symmetrization of homogeneous
polynomials, we can arbitrarily choose permutations $\theta_{1}\in
S_{p}$ and $\theta_{2}\in S_{q}$ such that
\begin{eqnarray}
\fl \Lambda\left(  f\right)    & =\frac{1}{p!}\sum_{\sigma\in\mathcal{S}_{p}%
}a^{j_{1}..,j_{p}}X_{\sigma\left(  j_{1}\right) }..X_{\sigma\left(
j_{p}\right)  }=\left( a^{j_{1}..,j_{p}}X_{\theta_{1}\left(
j_{1}\right)  }..X_{\theta_{1}\left(  j_{p}\right)  }+L.O.T.\right)  ,\nonumber\\
\fl \Lambda\left(  g\right)    & =\frac{1}{q!}\sum_{\sigma\in\mathcal{S}_{q}%
}b^{k_{1}..,k_{q}}X_{\sigma\left(  k_{1}\right) }..X_{\sigma\left(
k_{q}\right)  }=\left( b^{k_{1}..,k_{q}}X_{\theta_{2}\left(
k_{1}\right) }..X_{\theta_{2}\left(  k_{q}\right)
}+L.O.T.\right).\label{Sim7}
\end{eqnarray}
Expanding the commutator $\left[  \Lambda\left(  f\right)
,\Lambda\left(  g\right)  \right]$ using these representatives, a
reasoning identical to that above leads to the identity
\begin{equation}
a^{j_{1}..,j_{p}}b^{k_{1}..,k_{q}}\left[  X_{\theta_{1}\left(
j_{1}\right) }..X_{\theta_{1}\left(  j_{p}\right)
},X_{\theta_{2}\left(  k_{1}\right) }..X_{\theta_{2}\left(
k_{q}\right)  }\right]  =0.\label{Sim8}
\end{equation}
These identities enable us to sum over all permutations of $S_{p}$
and $S_{q}$, which leads to the equation
\begin{equation}
\sum_{\sigma\in S_{p}}\sum_{\tau\in S_{q}}a^{j_{1}..,j_{p}}b^{k_{1}..,k_{q}%
}\left[  X_{\sigma\left(  j_{1}\right)  }..X_{\sigma\left(
j_{p}\right) },X_{\tau\left(  k_{1}\right)  }..X_{\tau\left(
k_{q}\right)  }\right]  =0.\label{Sim9}
\end{equation}
But observe that equation (\ref{Sim9}) is exactly the commutator
of $\Lambda\left( f\right)$ and $\Lambda\left(  g\right)$ if we
use their symmetric representatives, from which we conclude that
\begin{equation}
\left[  \Lambda\left(  f\right)  ,\Lambda\left(  g\right)  \right]
=0,\label{Sima1}
\end{equation}
showing that the labelling operators $F=\Lambda(f),\;
G=\Lambda(g)$ commute.
\end{proof}

\smallskip
We observe that, in the frame of Berezin brackets, equation
(\ref{BeFa}) inherits meaning as a special case corresponding to
linear polynomials. Therefore the analytical approach to compute
Casimir operators can be seen as a particular application of the
preceding result.
\smallskip

The main interest of the previous theorem relies in its
application to the MLP, in order to check that two or more
labelling operators commute with each other. Let
$\frak{s}\supset\frak{s} ^{\prime}$ be an algebra-subalgebra
chain, where both $\frak{s}$ and $\frak{s} ^{\prime}$ are
reductive Lie algebras.

\begin{corollary}
Let $F=\Lambda(f),\;G=\Lambda(g)\in\mathcal{U}\left(
\frak{s}\right)  $ be two non-factorizable missing label
operators. Then $\left[  F,G\right]  =0$ if and only if $\left\{
f,g\right\}  =0$.
\end{corollary}

The requirement of reductive Lie algebras is imposed to guarantee
the existence of a complete basis of labelling operators formed by
polynomials \cite{Pe}. In order to complete the analytical
characterization, it should be justified that labelling operators
are always non-factorizable, in order to prevent the degenerate
case studied in the previous section. To this extent, suppose that
$P,Q$ is a pair of labelling operators such that $\left[
P,Q\right] \neq0$ and $\left\{ \pi\left(  P\right) ,\pi\left(
Q\right) \right\} =0$. By Lemmas 1 and 2, $P$ and $Q$ would have
the following shape:
\begin{eqnarray}
P  & =X_{\alpha_{1}^{1},}...X_{\alpha_{p}^{1}}P_{1}+X_{\alpha_{1}^{2}%
}..X_{\alpha_{q}^{2}}P_{2}+...X_{\alpha_{1}^{r}}..X_{\alpha_{s}^{r}}%
P_{r}+P^{\prime},\nonumber\\
Q  & =X_{\alpha_{p}^{1},}...X_{\alpha_{1}^{1}}Q_{1}+X_{\alpha_{q}^{2}%
}..X_{\alpha_{1}^{2}}Q_{2}+...X_{\alpha_{s}^{r}}..X_{\alpha_{1}^{r}}%
Q_{r}+Q^{\prime},\label{Do1}
\end{eqnarray}
where $\left[  P^{\prime},Q^{\prime}\right]  =\left[
P^{\prime},Q_{i}\right] =\left[  P_{i},Q^{\prime}\right]  =\left[
P_{i},Q_{i}\right]  =0$ for $i=1,..,r$ \ and \ $\left[
X_{\alpha_{t}^{l}},P_{k}\right]  =\left[
X_{\alpha_{t}^{l}},Q_{k}\right]  =0$ for $l=1,..,r;\,k=1,..,r$. As
a consequence, the corresponding analytical counterparts
$\pi\left( P\right) ,\pi\left( Q\right)  \in\mathcal{S}\left(
\frak{g}\right) $ would satisfy the relations
\begin{eqnarray}
\fl \left\{  \pi\left(  P^{\prime}\right)  ,\pi\left(
Q^{\prime}\right) \right\}    & =\left\{  \pi\left(  P_{i}\right)
,\pi\left(  Q_{i}\right) \right\}  =\left\{  \pi\left(
P^{\prime}\right)  ,\pi\left(  Q_{i}\right) \right\}  =\left\{
\pi\left(  P_{i}\right)  ,\pi\left(  Q^{\prime}\right)
\right\}  =0,\nonumber\\
\fl \left\{  x_{\alpha_{t}^{l}},\pi\left(  P_{k}\right)  \right\}
& =\left\{ x_{\alpha_{t}^{l}},\pi\left(  Q_{k}\right)  \right\}
=0,\label{Do2}
\end{eqnarray}
for $i=1,..r;\;l=1,..,r;\,k=1,..,r$. On the other hand, both
$\pi\left( P\right)  $ and $\pi\left(  Q\right)  $ are
functionally independent solutions to the subsystem of
differential equations (\ref{sys}) corresponding to the generators
of $\frak{s}^{\prime}$, thus subjected to additional constraints.
In practice, such degenerate independent labelling operators have
never been observed \cite{El,Sh2,Za,Que,Hu1,Is,C72,C75,GaK}, and
the large amount of conditions (\ref{Do1}) that such operators
should satisfy makes it unlikely that they are functionally
independent. In any case, since the constraints (\ref{Do2}) are
not required by the reduction chain, we can always find
independent labelling operators that are non-factorizable.

\subsection{Example: The chain $\frak{sp}(6)\supset\frak{su}(3)\times
\frak{u}(1)$}

The unitary reduction of the non-compact symplectic Lie algebra
$\frak{sp}(6,\mathbb{R})$ has found applications in the nuclear
collective model \cite{RoG}, where nuclear states are classified
as bases of irreducible representations of $\frak{sp}(6)$ reduced
with respect to $\frak{su}(3)\times \frak{u}(1)$. In this case, we
have to add $n=3$ labelling operators to distinguish the states.
Generating functions for this chain were studied in \cite{GaK},
which in particular allowed to derive an integrity basis
consisting of $31$ elementary subgroup scalars. However, the
problem of extracting three commuting operators was not
undertaken. We show that combining the decomposition of Casimir
operators of \cite{C75} with the previous results, such a set of
labelling operators can be obtained naturally.

\smallskip
To this extent, we use the Racah realization for the symplectic
Lie algebra $\frak{sp}\left(6,\mathbb{R}\right)$ \cite{Ra}. We
consider the generators $X_{i,j}$ with $-3\leq i,j\leq 3$
satisfying the condition
\begin{equation}
X_{i,j}+\varepsilon_{i}\varepsilon_{j}X_{-j,-i}=0,
\end{equation}
where $\varepsilon_{i}={\rm sgn}\left( i\right) $. Over this
basis, the brackets are given by
\begin{equation}
\left[  X_{i,j},X_{k,l}\right]  =\delta_{jk}X_{il}-\delta_{il}X_{kj}%
+\varepsilon_{i}\varepsilon_{j}\delta_{j,-l}X_{k,-i}-\varepsilon
_{i}\varepsilon_{j}\delta_{i,-k}X_{-j,l}, \label{Kl3}
\end{equation}
where $-3\leq i,j,k,l\leq 3$. Using the polynomial block matrix
$M$ defined by
\begin{equation}
M=\left(
\begin{array}
[c]{cc}%
x_{i,j} & \sqrt{-1}x_{-i,j}\\
\sqrt{-1}x_{i,-j} & -x_{i,j}%
\end{array}
\right)  \label{M0}%
\end{equation}
and computing the coefficients of the characteristic polynomial
\cite{C46}:
\begin{equation}
\left|  M-T\mathrm{Id}_{6}\right|  =T^{6}+C_{2}T^{4}+C_{4}T^{2}+C_{6}%
,\label{LOL}%
\end{equation}
we obtain three independent invariants $C_{2},C_{4}$ and $C_{6}$
of $\frak{sp}\left(  6,\mathbb{R}\right)  $, and the symmetrized
operators $\Lambda(C_{i})$ give the usual Casimir operators in the
enveloping algebra. As the unitary algebra $\frak{u}(3)$ is
generated by $\left\{  X_{i,j}|1\leq i,j\leq3\right\}  $, it
suffices to replace the diagonal operators $X_{i,i}$ by suitable
linear combinations to obtain a $\frak{su}(3)\times\frak{u}(1)$
basis. Taking $H_{1}=X_{1,1}-X_{2,2},\;H_{2}=X_{2,2}-X_{3,3}$ and
$H_{3}=X_{1,1}+X_{2,2}+X_{3,3}$ we get the Cartan subalgebra of $\frak{su}%
(3)$, while $H_{3}$ commutes with all $X_{i,j}$ with positive
indices $i,j$. The invariants over the
$\frak{su}(3)\times\frak{u}(1)$ basis are obtained replacing the
variables $x_{i,i}$ by the corresponding linear combinations of
$h_{i}$ \cite{C46}. Now we construct labelling operators for this
chain using the contraction method developed in \cite{C72}. The
transformations determined by
\begin{equation}
H_{i}^{\prime}=H_{i},\;X_{i,j}^{\prime}=X_{i,j},\;X_{-i,j}^{\prime}=\varepsilon
X_{-i,j},\;X_{i,-j}^{\prime}=\varepsilon X_{i,-j}
\end{equation}
define a contraction of $\frak{sp}(6)$ onto the inhomogeneous
algebra
$(\frak{su}(3)\times\frak{u}(1))\overrightarrow{\oplus}_{R}12L_{1}$,
where the representation $R$ decomposes into a sextet and
antisextet with $\frak{u}(1)$ weight $\pm1$ and a singlet with
$\frak{u}(1)$ weight $1$. As shown in \cite{C75}, this contraction
induces a decomposition of the Casimir operators of
$\frak{sp}(6)$.\footnote{The quadratic operator is skipped since
it provides no independent labelling operators \cite{C72,C75}.}
The rescaled Casimir operators can be written as:
\begin{equation}
\begin{array}[l]{ll}
C_{4}= \varepsilon^{4} C_{(4,0)}+ \varepsilon^{2} C_{(2,2)}+ C_{(0,4)},& \\
C_{6}= \varepsilon^{6} C_{(6,0)}+ \varepsilon^{4} C_{(4,2})+
\varepsilon^{2} C_{(2,4)}+ C_{(0,6)},&
\end{array}
\end{equation}
where $C_{(k,l)}$ denotes a homogeneous polynomial of degree $k$
in the variables of $R$ and degree $l$ in the variables of the
unitary subalgebra. It can be shown that the $C_{(i,j)}$ are
labelling operators \cite{C75}, and that the $C_{(0,k)}$, being
functions of the Casimir operators of $\frak{su}(3)\times
\frak{u}(1)$, do not provide independent labelling operators. The
symmetrized operators of $C_{(2,2)},C_{(4,2)}$ and $C_{(2,4)}$ are
added to the Casimir operators of $\frak{sp}(6)$ and the
subalgebra $\frak{su}(3)\times \frak{u}(1)$, and the resulting
nine operators can be easily seen to be functionally independent.
It is straightforward but time consuming to verify that these
operators are non-factorizable. To solve the MLP, it remains to
check the commutativity of the labelling operators. In order to
determine whether the symmetrization of $C_{(2,2)}$,$C_{(4,2)}$
and $C_{(2,4)}$ commute, we compute the Berezin bracket. A routine
but tedious computation shows that
\begin{equation}
\begin{array}
[c]{llll}%
\left\{ C_{(2,2)},C_{(4,2)}\right\}=0, & \left\{
C_{(2,2)},C_{(2,4)}\right\} =0, & \left\{
C_{(2,4)},C_{(4,2)}\right\} =0. &
\end{array}
\end{equation}
By Theorem 1, we conclude that
\begin{equation}
\fl
\begin{array}
[c]{llll}%
\left[ \Lambda(C_{(2,2)}),\Lambda(C_{(4,2)})\right]=0, & \left[
\Lambda(C_{(2,2)}),\Lambda(C_{(2,4)})\right]=0, & \left[
\Lambda(C_{(2,4)}),\Lambda(C_{(4,2)})\right]=0. &
\end{array}
\end{equation}
We remark that a direct evaluation of the commutators of these
symmetrized operators is a quite demanding computational problem,
since the polynomials $C_{(2,2)}$, $C_{(2,4)}$ and $C_{(4,2)}$
have $126$, $686$ and $444$ terms, respectively.

\section*{Summary and outlook}

We have shown that the commutativity of labelling operators in the
missing label problem can be solved using the analytical approach,
by means of the so-called Berezin bracket, up to a special type of
polynomials that is unlikely to appear in applications. We stress
that none of the known labelling problems does admit solutions of
such special type, and we firmly believe that do not appear as
labelling operators, at least for reductive Lie algebras. However,
even if such degenerate operators were possible, we can always
find independent labelling operators that are non-factorizable, to
which the analytical criterion is then applied. This result
constitutes a natural enlargement of the classical analytical
approach \cite{Pe}, and provides a general criterion to check the
commutativity of labelling operators without being forced to
determine first an integrity basis. Although the verification that
two operators are non-factorizable is formally very simple, it can
take a large amount of time depending on the number of terms.

\smallskip

The possibility of testing the commutativity of labelling
operators in the analytical frame also opens new possibilities for
the systematic search of the most general solution to a MLP.
Potential applications of this procedure concern reduction chains
for exceptional groups, as well as other high rank groups used in
high energy physics \cite{Wb}, where an approach by means of
enveloping algebras presents many computational problems. An
analysis of these labelling problems systematized along these
lines is currently in progress.

\section*{Acknowledgment}
The authors express their gratitude to J C Moreno and the referee
for useful discussions and valuable suggestions. During the
preparation of this work, the authors were financially supported
by the research projects FPA2006-02315 of the CICYT (LJB),
MTM2006-09152 of the M.E.C. and CCG07-UCM/ESP-2922 of the
U.C.M.-C.A.M. (RCS).

\end{document}